\documentclass[11pt]{article}

\usepackage{amsmath,amssymb,amsfonts}
\usepackage{fullpage}
\usepackage{url}
\usepackage{multirow}

\newtheorem{theorem}{Theorem}
\newtheorem{lemma}{Lemma}
\newtheorem{proposition}{Proposition}
\newtheorem{corollary}{Corollary}
\newtheorem{definition}{Definition}

\newtheorem{remark}{Remark}

\newtheorem{hypothesis}{Hypothesis}

\usepackage{float}

\usepackage{hyperref}
\hypersetup{
    unicode=false,          
    pdftoolbar=true,        
    pdfmenubar=true,        
    pdffitwindow=true,      
    pdftitle={On the Certification of the Restricted Isometry Property},    
    pdfauthor={Pascal Koiran, Anastasios Zouzias},     
    pdfsubject={},   
    pdfcreator={Pascal Koiran},   
    pdfproducer={}, 
    pdfkeywords={compressed sensing, restricted isometry property, basis pursuit}, 
    pdfnewwindow=true,      
    colorlinks=false,       
    linkcolor=red,          
    citecolor=green,        
    filecolor=magenta,      
    urlcolor=cyan          
}

\usepackage{algorithm}
\usepackage{algpseudocode}

\makeatletter
\def\@yproof[#1]{\@proof{ #1}}
\def\@proof#1{\begin{trivlist}\item[]{\em Proof#1.}}
\newenvironment{proof}{\@ifnextchar[{\@yproof}{\@proof{} 
}}{~$\Box$\end{trivlist}}
\makeatother

\def\const{\gamma}
\newcommand{\ignore}[1]{}
\newcommand{\abs}[1]{\left|#1\right|}
\newcommand{\RR}{\mathbb{R}}

\newcommand{\EE}{\mathbb{E}}

\newcommand{\norm}[1]{\ensuremath{\left\|#1\right\|}}

\newcommand{\frobnorm}[1]{\ensuremath{\left\|#1\right\|_{\text{\rm F}}}}

\newcommand{\trace}[1]{\ensuremath{\mathrm{\textbf{tr}}\left(#1\right)}}

\newcommand{\ip}[2]{\left\langle {#1},\ {#2} \right\rangle}

\newcommand{\OO}{\mathcal{O}}
\newcommand{\Id}{\mathbf{I}}
\newcommand{\J}{\mathbf{J}}

\newcommand\cc{\ensuremath{\mathbb{C}}}
\newcommand\rr{\ensuremath{\mathbb{R}}}
\newcommand\nn{\ensuremath{\mathbb{N}}}

\newcommand\p{\ensuremath{\mathsf P}}
\newcommand\np{\ensuremath{\mathsf{NP}}}

\newcommand\sk{\widehat{A}}
\newcommand\sign{A}

\title{On the Certification of the Restricted Isometry Property}

\author{%
Pascal Koiran\thanks{LIP, UMR 5668, ENS de Lyon -- CNRS -- UCBL -- INRIA, \'Ecole Normale Sup\'erieure de Lyon, Universit\'e de Lyon and Department of Computer Science, University of Toronto, Canada {\tt Pascal.Koiran@ens-lyon.fr}} \and Anastasios Zouzias\thanks{Department of Computer Science, University of Toronto, Canada {\tt zouzias@cs.toronto.edu}}
}

\begin{document}

\maketitle

\begin{abstract}
Compressed sensing is a technique for finding sparse solutions to underdetermined linear systems. This technique relies on properties of the sensing matrix such as the \emph{restricted isometry property}. Sensing matrices that satisfy 
this property with optimal parameters are mainly obtained via probabilistic arguments. Given any matrix, deciding whether it satisfies the restricted isometry property is a non-trivial computational problem. In this paper, we give reductions from dense subgraph problems to the certification of the restricted isometry property. This gives evidence that certifying 
this property is unlikely to be feasible in polynomial-time. Moreover, on the positive side we propose an improvement on the brute-force enumeration algorithm for checking the restricted isometry property.

Another contribution of independent interest is a spectral algorithm for
certifying that a random graph does not contain any dense $k$-subgraph.
This ``skewed spectral algorithm'' performs better than the basic spectral
algorithm in a certain range of parameters.
\end{abstract}

%

%
\section{Introduction}
%
Let $\Phi$ be a $n \times N$ matrix with $N \geq n$. 
A vector $x \in \cc^N$ is said
to be $k$-sparse if it has at most $k$ nonzero coordinates. 
Given $\delta \in ]0,1[$, $\phi$ is said to satisfy the Restricted Isometry
Property (RIP) of order $k$ with parameter $\delta$ if it approximately preserves the Euclidean norm in the following sense:
for every $k$-sparse vector $x$, we have
$$(1-\delta)||x||^2 \leq ||\Phi x||^2 \leq (1+\delta)||x||^2.$$
Clearly, for this to be possible we must have $k \leq n$.
Given $\delta$, $n$ and $N$, the goal is to construct RIP matrices with
$k$ as large as possible.
This problem is motivated by its applications to compressed sensing:
it is known from Cand\`es, Romberg and Tao~\cite{Candes,CRT06,CandesTao} that 
the restricted isometry property enables the efficient recovery of sparse signals using linear programming techniques.
For that purpose one can take any fixed $\delta < \sqrt{2}-1$~\cite{Candes}.

 Various probabilistic models are known to generate random matrices 
that satisfy the RIP with a value of  $k$ which is (almost) linear $n$.
See for instance Theorem~2 in Section~\ref{lazy} for the case of
matrices with entries that are independent symmetric $(\pm 1)$  
Bernouilli matrices. The recent survey~\cite{Vershynin} provides
additional results of this type and extensive references to the probabilistic
literature.
Some significant effort has been devoted to the construction of explicit
(rather than probabilistic) RIP matrices, but this appears to be a 
difficult problem. As pointed out by Bourgain et al. in 
a recent paper~\cite{BDFKK,BDFKKb}, most of the known explicit
constructions~\cite{Kashin75,AGHP92,Devore07}
are based on the construction of systems of unit vectors
with a small coherence parameter (see section~\ref{order} for a definition of this parameter and its connection to the RIP).
Unfortunately, this method cannot produce 
RIP matrices of order $k > \sqrt{n}$~\cite{BDFKK,BDFKKb}.
Bourgain et al. still manage to break through the $\sqrt{n}$
``barrier'' using techniques from additive combinatorics: they construct
RIP matrices of order $k=n^{1/2+\epsilon_0}$ where $\epsilon_0 > 0$ is an
unspecified ``explicit constant''. Note that this is still far from the 
order achieved by probabilistic constructions.
Here we study the restricted isometry property from the point
of view of computational complexity: what is the complexity of deciding
whether a matrix satisfies the RIP, and of computing or approximating
its order $k$ or its RIP parameter $\delta$? 
An efficient (deterministic) algorithm would have applications
to the construction of RIP matrices. One would draw a random matrix $\Phi$
from one of the well-established probabilistic models mentioned above,
and run this hypothetical algorithm on $\Phi$ to compute or approximate
$k$ and $\delta$. The result would be a matrix with {\em certified}
restricted isometry properties (see Section~\ref{lazy} for an actual result
along those lines). This may be the next best thing short of
an explicit construction (and as mentioned above, the known  explicit constructions are far from optimal).
The definition of the restricted isometry property suggests an exhaustive
search over $\binom{N}{k}$ subspaces, but we are not aware of any existing
hardness result suggesting that exhaustive
search is somehow unavoidable.
There has been more work from the algorithm design side. 
In particular, it was shown that semi-definite programming can be used
to verify the restricted isometry property~\cite{ABG08} and other related properties 
from compressed sensing~\cite{AspreG08,IouNem08}. Unfortunately, 
as pointed out in~\cite{AspreG08} these methods are unable to certify
the restricted isometry property
 for $k$ larger than $O(\sqrt{n})$, even for matrices that satisfy the RIP up to order $\Omega(n)$. As we have seen,
$k=O(\sqrt{n})$ is also the range where coherence-based methods reach their limits.
In this paper we provide both positive and negative results on the computational complexity of the RIP, 
including the range $k > \sqrt{n}$.
%
\subsection*{Positive Results}
%
In Section~\ref{order}, we study the relation between the RIP parameters
of different orders for a given matrix $\Phi$. Very roughly, we show
in Theorem~\ref{m2k} 
that the RIP parameter is at most proportional to the order.
We therefore have a trade-off between order and RIP parameter:
in order to construct a matrix of given order and RIP parameter, it suffices
to construct a matrix of lower order and smaller RIP parameter.
We illustrate this point in Section~\ref{lazy}. Our starting point
is the above-mentioned (very naive) exhaustive search algorithm,
which enumerates all $\binom{N}{k}$ subspaces generated by $k$ 
column vectors. 
We obtain  a ``lazy algorithm'' which enumerates instead all subspaces
generated by $l$ basis vectors for some $l < k$.
We show that the lazy algorithm can go slightly beyond the $\sqrt{n}$
barrier if a quasi-polynomial running time is allowed.
%
\subsection*{Negative Results}
%
We provide the first hardness results on the certification of the restricted
isometry property. We could not base these results on a well-established 
assumption such as $\p \neq \np$. Instead, we rely on certain assumptions
on the complexity of dense subgraph problems. There is already
a fairly large literature on the complexity of detecting dense subgraphs
in an input graph,
with both algorithmic results and hardness results (see~\cite{BCCFV10} for
extensive references).
In particular, it was shown by Khot~\cite{Khot06} that the dense $k$-subgraph 
problem does not have a 
polynomial time approximation scheme under the assumption that $\np$ does not have
randomized algorithms that run in sub-exponential time.
Earlier, Feige~\cite{Feige02} had obtained this result under the assumption
that random 3-SAT formulas are hard to refute.
Our results are of a similar flavor as those of Feige~\cite{Feige02} and Alekhnovich~\cite{Alekh03} since they are based
on average-case assumptions. However, 
dense subgraphs do not appear 
at the end of the reduction as in~\cite{Feige02} but at the beginning.
Our starting point is the assumption that it is hard to certify that a random $G(n,1/2)$ graph does not contain any $k$-subgraph with density at least $1/2+\epsilon$
(note that $1/2$ is the expected density). 
A certification algorithm must certify most of the graphs on $n$ vertices, and
any certified graph cannot contain any $k$-subgraph with density at least $1/2+\epsilon$.
More precisely, 
 we present two distinct
hardness hypotheses corresponding to different settings for the
parameters $k$ and~$\epsilon$. 
In the first hypothesis we set $k=n^{\alpha}$ and 
$\epsilon=1/n^{\beta}$, where $\alpha,\beta \in ]0,1[$ and $\alpha-\beta<1/2$.
In the second hypothesis we set $k=\alpha n$ and $\epsilon = \const /\sqrt{n}$
where $\alpha \in ]0,1[$
is a small enough constant, 
$\const>0$ and $\alpha \const$ is small enough.
Here, ``small enough'' means that 
$\alpha$ and $\alpha \const$ are both upper bounded by $\kappa$,
where 
$\kappa<1$ is a universal constant. The second hypothesis postulates
that with the above parameter settings, there is a value of $\kappa$ for which 
certifying most graphs on $n$ vertices is computationally difficult.
We believe that these two hypotheses are
consistent with the current state of knowledge on dense subgraph problems.
In particular, without the constraints on $\alpha,\beta$ and $\const$ these
hypotheses would be refuted by a simple spectral algorithm.
Each of the two hypotheses leads to different hardness results.
From the first hypothesis it follows that RIP parameters cannot be
approximated to within any constant factor in polynomial time.
From the second hypothesis we derive the following stronger result:
no polynomial time algorithm can distinguish between a matrix which
satisfies the RIP of order $k$ with parameter (say) $\kappa/1000$ and 
a matrix which does {\em not} satisfy the RIP with parameter $\kappa/2$
(here the constant 1000 can be replaced by any constant larger than 2).
%

%
Finally, in an effort to further test the validity of  these two hypotheses 
we propose a {\em skewed spectral algorithm} for certifying the absence 
of a $k$-subgraph of  density at least $1/2+\epsilon$ in a random graph.
This algorithm performs better than the  basic spectral algorithm 
for certain 
settings of $k$ and $\epsilon$, 
but is still consistent with our two hypotheses.
%
\subsection*{Organization of the paper}
%
As explained above, the next two sections are devoted to positive results.
In Section~\ref{eigenvalues} we work out some bounds on the eigenvalues of
random matrices, for later use in our reductions from dense subgraph problems
to the approximation of RIP parameters.
We rely mainly on the classical work of F\"uredi and Koml\'os~\cite{FK81}
as well as on a more recent concentration inequality due to Alon, Krivelevich
and Vu~\cite{AKV02}.
In Section~\ref{subgraphs} we present our two hypotheses on the intractability
of dense subgraph problems, and we show that they are consistent with
what we know from the spectral algorithm.
In Section~\ref{hard_square} we use these hypotheses to show that 
approximating RIP parameters is hard even for square matrices.
In Section~\ref{hard_rect} we derive similar results for matrices of
``strictly rectangular'' format (which is the case of interest in compressed
sensing).
We proceed by reduction from the square case.
Interestingly, this last reduction relies on the known constructions
(deterministic~\cite{BDFKK,BDFKKb} and probabilistic~\cite{Vershynin}) 
of matrices with good RIP parameters
mentioned earlier in the introduction.
We therefore turn these positive results into negative results.
The table at the end of Section~\ref{hard_rect} gives a summary of 
our hardness results.
Section~\ref{sec:skewed_alg} is devoted to our skewed spectral algorithm. This part
of the paper can be read independently from the RIP results as it deals
with a purely graph-theoretic problem. The analysis of the skewed 
algorithm is based
on a new spectral norm estimate for a certain class of random matrices.
\section{Increasing the order by decreasing the RIP parameter}
\label{order}
As explained at the beginning of~\cite{BDFKK,BDFKKb}, 
certain (suboptimal) constructions
are based on the construction of systems of unit vectors 
$(u_1,\ldots,u_N) \in \cc^n$ with small coherence. The coherence parameter
$\mu$ is defined as $\max_{i \neq j} |\langle u_i, u_j \rangle|$.
Indeed, we have the following proposition.
\begin{proposition} \label{coherence}
Assume that the column vectors $u_1,\ldots,u_N$ of $\Phi$ are of norm~1
and coherence $\mu$. Then $\Phi$ satisfies the RIP of order $k$ with
parameter $\delta=(k-1)\mu$.
\end{proposition}
We reproduce the proof from~\cite{BDFKK,BDFKKb} since if fits in one line:
for any $k$-sparse vector $x$, 
$$| ||\Phi x||^2 - ||x||^2| \leq 2 \sum_{i<j} |x_ix_j \langle u_i,u_j \rangle| \leq \mu ((\sum_i |x_i|)^2 - ||x||^2) \leq (k-1)\mu||x||^2.$$
%

%
We now give a result, which (as we shall see) generalizes Proposition~\ref{coherence}.
\begin{theorem} \label{m2k}
Assume that $\Phi$ has unit column vectors and satisfies the RIP of order $m$ with parameter $\epsilon$. For $k \geq m$, $\Phi$ also satisfies the RIP
of order $k$ with parameter $\delta = \epsilon(k-1)/(m-1)$.
\end{theorem}
\begin{proof}
Let $u_1,\ldots,u_N$ be the column vectors of $\Phi$. Let $x$ be a $k$-sparse
vector, and 
write $x=\sum_{i \in T} x_iu_i$ where $T$ is a subset of $\{1,\ldots,N\}$ of
size $k$.
Since $||\Phi x||^2 = ||x||^2 + 2 \sum_{i<j} x_ix_j \langle u_i,u_j \rangle$, to check the RIP of order $k$ we need to show that
\begin{equation} \label{innerprods}
|\sum_{i<j} x_i x_j \langle u_i,u_j \rangle| \leq \delta||x||^2/2,
\end{equation}
where $\delta = \epsilon(k-1)/(m-1)$.
To estimate the left hand side, we compare it to the sum of 
the similar quantity taken over all subsets of size $m$ of $T$, namely:
\begin{equation} \label{sum}
|\sum_{|S|=m} \sum_{i,j \in S, i<j} x_i x_j \langle u_i,u_j \rangle|.
\end{equation}
Since each pair $(i,j)$ appears in exactly $\binom{k-2}{m-2}$ subsets of 
size $m$, this sum is equal to $\binom{k-2}{m-2}$ times the left-hand side
of~(\ref{innerprods}). But we can also estimate (\ref{sum}) using the RIP
of order $m$. 
For each subset $S$ of size $m$, we have 
$$|\sum_{i,j \in S, i<j} x_i x_j \langle u_i,u_j \rangle| \leq \epsilon\sum_{i \in S}x_i^2/2.$$
This follows from~(\ref{innerprods}), replacing $\delta$ 
by $\epsilon$ (the RIP parameter of order $m$).
Since each term $x_i^2$ will appear exactly in $\binom{k-1}{m-1}$ subsets,
we obtain $\epsilon \binom{k-1}{m-1}||x||^2/2$ 
as an upper bound for~(\ref{sum}).
We conclude that  the left-hand side
of~(\ref{innerprods}) is bounded by $\frac{\epsilon}{2} \binom{k-1}{m-1}||x||^2/\binom{k-2}{m-2} = 
\epsilon \frac{k-1}{m-1}||x||^2/2$.\end{proof}
We claim that Proposition~\ref{coherence} is the case $m=2$ of 
Theorem~\ref{m2k}. This follows from the following observation.
\begin{remark} \label{order2}
For a matrix $\Phi$ with unit column vectors, the coherence parameter $\mu$
is equal to the RIP parameter of order 2.
\end{remark}
\begin{proof}
Let $\delta$ be the RIP parameter of order 2. We have $\delta \leq \mu$ by Proposition~\ref{coherence}. It remains to show that $\delta \geq \mu$.
Consider therefore two column vectors $u_i$ and $u_j$ with 
$|\langle u_i,u_j \rangle| = \mu$. 
Let $x=u_i+u_j$. We have $||x||^2=2$ and $||\Phi x||^2=2 \pm 2\mu$, so 
that $\delta \geq \mu$ indeed.
\end{proof}
%
\section{A Matrix Certification Algorithm} \label{lazy}
%
%


%
%
The naive algorithm for computing the RIP parameter of order $k$ will involve the enumeration of the $\binom{N}{k}$ submatrices of $\Phi$ made up
of $k$ column vectors of $\Phi$. For each $T \subseteq \{1,\ldots,N\}$ of
size $k$ let us denote by $\Phi_T$ the corresponding $n \times k$ matrix.
We need to compute (or upper bound) $\delta = \max_T \delta_T$, where 
$$\delta_T = \sup_{x \in \cc^k} |\ ||\Phi_T x||^2 /||x||^2 -1\ |.$$
For each $T$, $\delta_T$ can be computed efficiently by linear algebra. 
For instance, $\delta_T$ is the spectral radius of the self-adjoint matrix
$\Phi_T^* \Phi_T - \Id_k$.
The cost of the computation is therefore dominated by the combinatorial factor
$\binom{N}{k}$ due to the enumeration of all subsets of size $k$.
Here we analyze what the naive algorithm can gain from Theorem~\ref{m2k}.
We therefore consider the following {\em lazy algorithm.}
\begin{algorithm}{}
	\caption{}\label{alg:lazy}
\begin{algorithmic}[1]
\Procedure{Lazy}{$\Phi$, $m$, $\delta$}
\State {\bf Input:} a $n \times N$ matrix $\Phi$ with unit column vectors, an integer $m \leq n$,
and a parameter $\delta \in ]0,1[$.
\State Compute as explained above the RIP parameter of order $m$. Call it $\epsilon$.
\State {\bf Output:} Certify $\Phi$ as a RIP matrix of order $k$ with parameter $\delta$,
for all $k \geq m$ such that $\epsilon (k-1)/(m-1) \leq \delta$.
\EndProcedure 
\end{algorithmic}
\end{algorithm}
The correctness of the 
algorithm follows immediately from Theorem~\ref{m2k}. We now analyze its behavior on random matrices, which are in many cases known to satisfy the RIP with high probability. Consider for instance 
the case of a matrix whose entries are independent symmetric Bernouilli
random variables.
\begin{theorem} \label{randrip}
Let $A$ be a $n \times N$ matrix whose entries are  independent symmetric 
Bernouilli random variables and assume that 
$n \geq C \epsilon^{-2}m \log(eN/m)$.
With probability at least $1-2\exp(-c\epsilon^2n)$, 
the normalized matrix $\Phi=\frac1{\sqrt{n}}A$ 
satisfies the  the RIP of order $m$ with parameter $\epsilon$.
Here $C$ and $c$ are absolute constants.
\end{theorem}
In fact the same theorem holds for a very large class of random matrix models,
namely, subgaussian matrices with either independent rows or independent columns (\cite{Vershynin}, Theorem~64).
\begin{proposition} \label{lazy_analysis}
Let $A$ be a random matrix as in Theorem~\ref{randrip}, 
and $\delta \in ]0,1[$. 
With probability at least $1-2(eN/m)^{-cCm}$, 
the lazy algorithm presented above will certify 
that $A$ satisfies the RIP of order $k$ with parameter $\delta$ for
all $k$ such that: $$k \leq \delta \sqrt{\frac{mn}{c \log (eN/m)}}.$$
Here $c$ and $C$ are the absolute constants from Theorem~\ref{randrip}.
\end{proposition}
\begin{proof}
All parameters being fixed we take $\epsilon$ as small as allowed by 
Theorem~\ref{randrip}, so that $n\epsilon^2=Cm \log(eN/m)$.
This yields the announced probability estimate, and the upper bound on $k$ 
is $\delta m/\epsilon$.
\end{proof}
To compare the lazy algorithm to the naive algorithm, set for instance 
$m=\sqrt{n}$. In applications to compressed sensing one can set $\delta$ to
a small constant value (any $\delta < \sqrt{2}-1$ will do). Thus, disregarding constant and
logarithmic factors, with high probability the lazy algorithm will certify 
the RIP property for $k$ of order roughly $n^{3/4}$.
This is achieved by enumerating 
$\binom{N}{n^{1/2}}$ subspaces, whereas the naive algorithm would enumerate roughly $\binom{N}{n^{3/4}}$ subspaces.

Another choice of parameters in Proposition~\ref{lazy_analysis} shows that one can beat the $\sqrt{n}$ 
bound by a logarithmic factor
with a quasi-polynomial time algorithm.
For instance:
\begin{corollary}
If we set $m=(\log N)^3$, the lazy algorithm runs in time $2^{O(\log^4 N)}$
and, with probability at least $1-2^{-\Omega((\log N)^4)}$
 certifies that $A$ satisfies the RIP of order $k$ with parameter $\delta$
for all $k \leq K \delta \log N \sqrt{n}$, where $K$ is an absolute constant.
\end{corollary}
%
%

\section{Eigenvalues of Random Symmetric Matrices} \label{eigenvalues}
%
Proposition~\ref{modelC} is the main probabilistic inequality that we derive in this section.
It shows that square matrices obtained by Cholesky decomposition from a certain class of random matrices
have good RIP parameters with high probability. This result is then used in Section~\ref{hard_square} 
to give a reduction from dense subgraph problems to the approximation of RIP parameters.
%
\subsection{Model A}
%
Consider the following random matrix model: $A$ is a symmetric $k \times k$ 
matrix with $a_{ii}=0$, and for $i < j$ the $a_{ij}$ are independent symmetric
Bernouilli random variables.

Let $\lambda_1(A) \geq \lambda_2(A) \geq \dots \lambda_k(A)$ 
be the eigenvalues of $A$. Let $m_s$ be the median of $\lambda_s(A)$.
From the main result of~\cite{AKV02} (bottom of p.~$263$) we have for $t \geq 0$
the inequality:
\[\Pr[\lambda_s(A) - m_s \geq t] \leq 2e^{-t^2/32s^2}.\]
From F\"uredi and Koml\'os (\cite{FK81}, Theorem~2) we know that 
$m_1 \leq 3 \sigma \sqrt{k}$ for~$k$ large enough, 
where $\sigma=1$ is the standard deviation of the $a_{ij}$ in the case $i<j$.
Therefore we have 
\[\Pr[\lambda_1(A) \geq 3 \sqrt{k}+t] \leq 2e^{-t^2/32}.\]
Since $\lambda_k(A)=-\lambda_1(-A)$ and $-A$ has same distribution as $A$,
we also have 
\[\Pr[\lambda_k(A) \leq -3 \sqrt{k}-t] \leq 2e^{-t^2/32}\]
(one could also apply directly the bound on $\lambda_k(A)$ for the more
general model considered in~\cite{AKV02}).
As a result:
\begin{proposition} \label{modelA}
There is an integer $k_0$ such that for all $k \geq k_0$ 
and for all $t \geq 0$ we have:
\[\Pr[ \max_i |\lambda_i(A)| \geq 3 \sqrt{k}+t] \leq 4e^{-t^2/32}.\]
\end{proposition}
\begin{remark} \label{constant3}
The constant 3 in Proposition~\ref{modelA} can be replaced
by any constant bigger than 2 (see Theorem~2 in~\cite{FK81}).
\end{remark}
%
\subsection{Model B}
%
Next we consider the model where $B$ is a symmetric $k \times k$ matrix satisfying the following condition: $b_{ii}=1$, and $b_{ij}=c \cdot a_{ij}/\sqrt{n}$ 
 for $i<j$, where the $a_{ij}$ are independent symmetric
Bernouilli random variables. Here $c>0$ is a fixed constant, and $n$ is an additional parameter 
which should be thought of as going to infinity with $k$.
\begin{corollary} \label{modelB}
Assume that $k \geq k_0$ 
and that 
$\delta\sqrt{n} \geq 3c\sqrt{k}$.
Then the eigenvalues of $B$ all lie in the interval
$[1-\delta,1+\delta]$ with probability at least
\[1-4\exp[-(\frac{\delta\sqrt{n}}{c} -3\sqrt{k})^2/32].\]
\end{corollary}
\begin{proof}
We have $B=\Id_k + cA/\sqrt{n}$, where $A$ follows the model of Proposition~\ref{modelA}. The result therefore follows from that proposition by choosing
$t$ so that $c(3\sqrt{k}+t)/\sqrt{n}=\delta$, i.e., 
$t=\delta\sqrt{n}/c-3\sqrt{k}$.
\end{proof}
In the next corollary we look at the case $n=k$ of this model.
\begin{corollary} \label{modelB'}
Assume  that $n \geq k_0$ 
and  $3c < 1$.
Then $B$ is positive semi-definite with probability at least
$$1-4\exp[-(1/c -3)^2 n/32].$$
\end{corollary} 
\begin{proof}
Set $n=k$ and $\delta=1$ in Corollary~\ref{modelB}.
\end{proof}
In the last result of this subsection we consider again the 
model $B=\Id_n + cA/\sqrt{n}$. 
Given a $n \times n$ matrix $M$ 
and two subsets $S,T \subseteq \{1,\ldots,n\}$ of size $k$, 
let us denote by $M_{S,T}$ the $k \times k$ sub-matrix made up of all entries
of $M$ of row number in $S$ and column number in $T$.
\begin{corollary} \label{modelB''}
Consider the random matrix 
$B=\Id_n + cA/\sqrt{n}$ where $A$ is drawn from
the uniform distribution on the set $n \times n$ symmetric matrices
with null diagonal entries and $\pm 1$ off-diagonal entries.

If $n \geq k \geq k_0$, then with probability at least 
 $$1-4\exp\left[k \ln (ne/k)-( \frac{\delta\sqrt{n}}{c} 
-3\sqrt{k})^2/32\right]$$
the submatrices $B_{S,S}$ have all their eigenvalues in the interval 
$[1-\delta,1+\delta]$ for all subsets $S\subseteq \{1,\ldots,n\}$ of size $k$.
\end{corollary}
\begin{proof}
By Corollary~\ref{modelB}, for each fixed $S$ 
matrix $B_{S,S}$ has an eigenvalue outside of the interval 
$[1-\delta,1+\delta]$ with probability at most
$4\exp[-( \frac{\delta\sqrt{n}}{c} -3\sqrt{k})^2/32].$
The result follows by taking a union bound over the 
$\binom{n}{k} \leq (ne/k)^k$ subsets of size $k$.
\end{proof}
%
\subsection{Model C} \label{sectionC}
%
In Corollaries~\ref{modelB'} and~\ref{modelB''}
 we considered the following random model for $B$:
set $B = \Id_n + c A/\sqrt{n}$, where $A$ is chosen from the uniform distribution on the set $S_n$ of all symmetric matrices with null diagonal entries
and $\pm 1$ off-diagonal entries.
If $B$ is positive semi-definite, we can find by Cholesky decomposition 
a $n \times n$
matrix $C$ such that $C^TC=B$.
If $B$ is not positive semi-definite, we set $C=0$.
This is the random model for $C$ that we study in this subsection.
\begin{proposition} \label{modelC}
Assume that $n \geq k \geq k_0$ 
and that 
$3c < \min(1, \delta\sqrt{n}/\sqrt{k})$.
With probability at least
 $$1-4\exp\left[k \ln (ne/k)-(\frac{\delta\sqrt{n}}{c} 
-3\sqrt{k})^2/32\right]-4\exp[-(1/c -3)^2 n/32],$$
$C$ satisfies the RIP of order $k$ with parameter $\delta$.
\end{proposition}
\begin{proof}
%
If $B = \Id_n + c A/\sqrt{n}$ is not positive semi-definite then $C=0$
and this matrix obviously does not satisfy the RIP.
By Corollary~\ref{modelB'}, $B$ can fail to be positive semi-definite
with probability at most $4\exp[-(1/c -3)^2 n/32].$
If $B$ is positive semi-definite then $C^TC=B$.
Using the notation of Corollary~\ref{modelB''}, 
matrix $C$ satisfies the RIP of order $k$ with parameter $\delta$ 
if for all subsets $S$ of size $k$,
the eigenvalues of the $k \times k$ matrices $(C^TC)_{S,S}$
all lie in the interval $[1-\delta,1+\delta]$.
Since $C^TC=B$, by Corollary~\ref{modelB''} this can happen with probability 
at most $4\exp[k \ln (ne/k)-( \frac{\delta\sqrt{n}}{c} 
-3\sqrt{k})^2/32].$
\end{proof}
%
\section{Dense Subgraph Problems} \label{subgraphs}
%
In this section we formulate precisely our two hypotheses on the intractability of dense subgraph problems.
We show that they are compatible with what follows from the spectral algorithm (see Theorem~\ref{spectral_algo} below).
Spectral methods have become classical in graph theory, see for instance~\cite{AKS98} for a nontrivial application to 
the detection of large hidden cliques in random graphs.
We work with the standard $G(n,1/2)$ random graph model: each of the 
$n(n-1)/2$ potential edges is chosen with probability $1/2$, and these choices
are made independently.
\begin{definition}
We say that a graph $H$ on $k$ vertices has density $\lambda$ if it has at least $\lambda k(k-1)/2$ edges. 
For $\epsilon \geq 0$, we say that $G$ has excess $\epsilon$ 
if it has density $1/2+\epsilon$.
\end{definition}
\begin{lemma}
In a random graph on $n$ vertices, the probability of existence of a
$k$-subgraph with excess $\epsilon$ is upper bounded by
$$\exp[k \ln (ne/k) -\epsilon^2k(k-1)].$$
\end{lemma}
\begin{proof}
Fix any subgraph $H$ on $k$ vertices. The number of edges in $H$ is a 
sum of $k(k-1)/2$ independent Bernouilli random variables.
Therefore, by Hoeffding's inequality the probability for $H$ to have excess $\epsilon$ can be upper bounded by
$\exp[-\epsilon^2k(k-1)].$
The result follows by taking a union bound over the 
$\binom{n}{k} \leq (ne/k)^k$ subgraphs of size $k$.
\end{proof}
This probability goes rapidly to 0 once $\epsilon$ becomes significantly
larger than $1/\sqrt{k}$. For instance:
\begin{corollary} \label{typical}
Fix any constant $c>0$. For $\epsilon \geq (1+c)\sqrt{\ln(ne/k)/k}$,
the probability that a random graph on $n$ vertices has a $k$-subgraph with excess $\epsilon$ goes to 0 as $k,n \rightarrow +\infty$.
\end{corollary}
When $\epsilon$ is in the range of Corollary~\ref{typical}, it makes sense
to consider algorithm which will certify that most graphs on $n$ vertices
do not have any $k$-subgraph with excess $\epsilon$ (i.e., a certified graph
should not have any  $k$-subgraph with excess $\epsilon$, and most 
graphs should be certified).

As we shall now see (in Theorem~\ref{spectral_algo}), the spectral method provides
an efficient certification algorithm if 
$\epsilon k / \sqrt{n}$ is sufficiently large.
\begin{lemma} \label{expansion}
Let $A$ be the signed adjacency matrix of a graph $G$.
If $G$ has a $k$-subgraph $H$ with excess $\epsilon$ then 
there is a unit vector $x \in \rr^n$ supported by $k$ basis vectors 
such that $x^TAx \geq 2\epsilon(k-1)$.
\end{lemma}
\begin{proof}
Assume without loss of generality that $V(G)=\{1,\ldots,n\}$ 
and $V(H)=\{1,\ldots,k\}$. 
We have
$a_{ii}=0$ and for $i \neq j$, $a_{ij}=1$ if $ij \in E$; 
$a_{ij}=-1$ if $ij {\not \in} E$.
Let $x \in \rr^n$ be the unit vector 
uniformly supported by the first $k$ basis vectors: 
$x_i=1/\sqrt{k}$ for $i \leq k$, $x_i=0$ for $i>k$.
Then $x^TAx=2\sum_{1 \leq i<j \leq k}a_{ij}/k \geq 2\epsilon(k-1)$. 
\end{proof}
\begin{corollary} \label{dense_lambda}
Under the same hypothesis as in Lemma~\ref{expansion}, 
the largest eigenvalue of $A$ satisfies $\lambda_1(A) \geq 2\epsilon(k-1)$.
\end{corollary}
\begin{proof}
This follows directly from the lemma since for a symmetric matrix,
$\lambda_1(A) = \sup_{||x||=1} x^TAx$.
\end{proof}
This leads to the following polynomial-time certification algorithm.
\begin{algorithm}{}
	\caption{}\label{alg:spectral}
\begin{algorithmic}[1]
\Procedure{Spectral}{$G$}
\State {\bf Input:} a graph $G$ on $n$ vertices.
\State Compute $\lambda_1(A)$, where $A$ is the signed adjacency matrix
of $G$.
\State {\bf Output:} Certify that $G$ does not contain any $k$-subgraph with excess $\epsilon$ for all $k,\epsilon$ such that $2\epsilon(k-1) > \lambda_1(A)$.
\EndProcedure 
\end{algorithmic}
\end{algorithm}
%
%
\begin{theorem} \label{spectral_algo}
The spectral algorithm presented above 
certifies that most graphs on $n$ vertices do not
contain any $k$-subgraph with excess $\epsilon$ for all $k,\epsilon$ such that
$$2\epsilon(k-1) \geq 3 \sqrt{n}.$$
\end{theorem}
\begin{proof}
The correctness of the algorithm follows from Corollary~\ref{dense_lambda}.
By Theorem~2 in~\cite{FK81} most graphs satisfy $\lambda_1(A) < 3\sqrt{n}$,
so most graphs are certified.
\end{proof}
\begin{remark} \label{constant3'}
The constant 3 in this theorem can be replaced by any constant larger than 2
(see also Remark~\ref{constant3}).
\end{remark}

In the case $\epsilon=1/2$, the spectral algorithm of Theorem~\ref{spectral_algo} certifies that most graphs on $n$ vertices do not
contain any $k$-clique for $k \geq 1+ 3\sqrt{n}$.

We put forward the following hypotheses.
\begin{hypothesis}\label{hyp:one}
Fix two arbitrary rational constants $\alpha, \beta \in ]0,1[$ such that
$\alpha-\beta < 1/2$. Set $\epsilon(n)=1/n^{\beta}$ and $k(n)=n^{\alpha}$ 
(we only consider values of $n$ such that $k(n) \in \nn$).

No polynomial time algorithm can certify that
most graphs on $n$ vertices do not have any $k(n)$-subgraph 
with excess $\epsilon(n)$.
\end{hypothesis}
\begin{hypothesis}\label{hyp:two}
There is a constant $\kappa \in ]0,1[$ such that the following holds: fix two arbitrary constants $\alpha \in ]0,1[$ and $\const >0$ with 
$\alpha \leq \kappa$ and 
$\alpha \const \leq \kappa$. 
Set $\epsilon(n)=\const /\sqrt{n}$ and $k(n)=\alpha n$ 
(again we only consider values of $n$ such that $k(n) \in \nn$).
Then no polynomial time algorithm can certify that
most graphs on $n$ vertices do not have any $k(n)$-subgraph 
with excess $\epsilon(n)$.
\end{hypothesis}
The constraints on the parameters $\alpha,\beta,\const$ in these two hypotheses
are meant to ensure consistency with what we know from the spectral algorithm:
in Hypothesis~\ref{hyp:one} we have $\epsilon k = n^{\alpha-\beta} < \sqrt{n}$. In Hypothesis~\ref{hyp:two} we have $\epsilon k = \alpha \const / \sqrt{n} \leq \kappa \sqrt{n} < \sqrt{n}$.
By Remark~\ref{constant3'}, if we allowed a value $\kappa > 1$ in Hypothesis~\ref{hyp:two}
then this hypothesis would become provably incorrect. 
Hypothesis~\ref{hyp:two} essentially means that the spectral algorithm can be improved
at most by a constant factor. As we will see later, both hypotheses are also consistent with the skewed algorithm of Section~\ref{sec:skewed_alg}.

Of course, our two  hypotheses are meaningful only for values of the parameters
for which it is actually the case 
that most graphs do not have a $k$-subgraph
with excess $\epsilon$ (otherwise, they are vacuously true).
By Corollary~\ref{typical}, for the first hypothesis it is enough 
to have $\beta < \alpha/2$. Therefore we can first pick any $\beta \in ]0,1/2[$ 
and then any $\alpha \in ]2\beta,\beta+1/2[$ (the upper bound $\beta+1/2$
enforces the constraint $\alpha-\beta < 1/2$ in Hypothesis~\ref{hyp:one}).

Let us now work out the corresponding constraints for Hypothesis~\ref{hyp:two}.
Take for instance $c=1$ in Corollary~\ref{typical}.
Plugging in $k = \alpha n$ and $\epsilon=\const/\sqrt{n}$ yields the condition
$$\const \geq 2\sqrt{\frac{\ln(e/\alpha)}{\alpha}}.$$
But we also have the constraint $\const \leq \kappa/\alpha$ from Hypothesis~\ref{hyp:two}.
So we should have
$$\kappa/\alpha \geq 2 \sqrt{ \frac{\ln(e/\alpha)}{\alpha}}.$$
This constraint is satisfied by all sufficiently small $\alpha$.
Hence we can first pick any small enough $\alpha$, and then any 
$\const$ in the interval
$[2\sqrt{\ln(e/\alpha)/\alpha},\kappa/\alpha].$
%
%
%
\section{A Skewed Spectral Algorithm}\label{sec:skewed_alg}
%
Let $G$ be an undirected graph on $n$ vertices and $\sign$ 
its signed adjacency matrix.
We define the skewed adjacency matrix of $G$ by the formula
\begin{equation} \label{skewdef}
\sk=\frac{\sign}{2}+\frac{a}{\sqrt{n}} \J, 
\end{equation}
where 
 $\J$  is the $n \times n$ all 1's matrix, and 
$a \geq 0$ a parameter to be tuned later. We call the matrix $\widehat{A}$ skewed because its entries are random variables with mean slightly above zero.
\begin{lemma} \label{eigendense}
If $G$ has a $k$-subgraph $H$ with excess $\epsilon$ then 
there is a unit vector $x \in \rr^n$ supported by $k$ basis vectors 
such that $x^T\sk x \geq \epsilon(k-1)+ka/\sqrt{n}$.
\end{lemma}
\begin{proof}
Assume without loss of generality that $V(G)=\{1,\ldots,n\}$ 
and $V(H)=\{1,\ldots,k\}$. 
Let $x \in \rr^n$ be the unit vector 
uniformly supported by the first $k$ basis vectors: 
$x_i=1/\sqrt{k}$ for $i \leq k$, $x_i=0$ for $i>k$.
We have seen that $x^T\sign x \geq 2\epsilon (k-1)$. 
The conclusion follows since $x^T \J x = k$.
\end{proof}
This simple lemma forms the basis for the following {\em skewed algorithm}.
\begin{algorithm}{}
	\caption{}\label{alg:skewed_spectral}
\begin{algorithmic}[1]
\Procedure{Skewed Spectral}{$G$, $a$}
\State {\bf Input:} a graph $G$ on $n$ vertices.
\State Compute $\lambda_1(\sk)$, where as above $\sk$ 
is the skewed adjacency matrix of $G$.
\State {\bf Output:} Certify that $G$ does not contain any $k$-subgraph with excess $\epsilon$ for all $k,\epsilon$ such that
$\epsilon(k-1) + ka/\sqrt{n} > \lambda_1(\sk)$.
\EndProcedure 
\end{algorithmic}
\end{algorithm}
In the case $a=0$, this is just the standard spectral algorithm.
Since $\sk$ is symmetric, $\lambda_1(\sk) 
= \sup_{\norm{x}=1} x^T\sk x$.
The correctness of the skewed algorithm therefore follows 
from Lemma~\ref{eigendense}. 
In the remainder of this 
section, we analyze the behavior of this algorithm 
on random graphs.
\begin{theorem} \label{norm}
Let $G$ be a random graph on $n$ vertices and $\sk$ its skewed 
adjacency matrix as defined by~(\ref{skewdef}).
The expectation of the 
operator norm
$\norm{\sk}=\sup_{\norm{x}=1} \norm{\sk x}$ satisfies the inequality
\[ \EE \norm{\sk}  \leq \sqrt{n} \sqrt{a^2 +5/4 + o(1)}.\]
Moreover, for all $\epsilon\geq 0$ we have the concentration inequality
\[\Pr\left[\norm{\sk} \geq \EE \norm{\sk} + \epsilon \right]\leq 
\exp(-\epsilon^2 /8).\]
\end{theorem}
We defer the proof of the above theorem to the appendix. Compare this bound 
on $\EE \norm{\sk}$
with the bound 
$(a+1+o(1))\sqrt{n}$ that follows from applying the triangle inequality on $\widehat{A}$. A direct calculation shows that 
the bound in Theorem~\ref{norm} is better for $a>1/8$.

%
For any matrix, the largest eigenvalue is upper bounded by the operator norm.
We therefore have the following corollary.
\begin{corollary} \label{skew_bound}
For most graphs $G$ on $n$ vertices, the skewed adjacency matrix of $G$
satisfies the inequality 
$$\lambda_1(\sk) \leq \sqrt{n} \sqrt{a^2 +5/4+o(1)}.$$
\end{corollary}
Now, let $k$ and $\epsilon$ be two functions of $n$.
From Corollary~\ref{skew_bound} and the correctness of the skewed algorithm
it follows that whenever the inequality
\begin{equation} \label{liminf}
\liminf_{n \rightarrow + \infty} \frac{1}{\sqrt{n}}\left[\epsilon(n).(k(n)-1)
+\frac{ak(n)}{\sqrt{n}}\right] > \sqrt{a^2+5/4}
\end{equation}
is satisfied, the skewed algorithm can certify that most graphs on $n$ vertices
do not have any $k(n)$-subgraph with excess $\epsilon(n)$.
In particular, we have obtained the following result.
\begin{theorem} \label{skth}
Set $\epsilon(n) = \const / \sqrt{n}$ and $k(n)=\alpha n$, 
where $\alpha \in ]0,1[$ and $\const > 0$ are two constants. 
If $\alpha \const + \alpha a > \sqrt{a^2+5/4}$ 
then the skewed algorithm can certify that most graphs on $n$ vertices
do not have any $k(n)$-subgraph with excess $\epsilon(n)$.
\end{theorem}
For instance, taking $a=1$ yields the following.
\begin{corollary} \label{aequals1}
If $\alpha \const + \alpha  > 3/2$ 
then the skewed algorithm can certify that most graphs on $n$ vertices
do not have any $\alpha n$-subgraph with excess $\const/\sqrt{n}$.
\end{corollary}
This is clearly an improvement over the standard spectral algorithm, since 
that algorithm is unable to certify that most graphs on $n$ vertices
do not have any $\alpha n$-subgraph with excess $\const/\sqrt{n}$
when $\alpha \const < 1$. But there are values $\alpha \in ]0,1[$ and
$\const > 0$ for which $\alpha \const < 1$ 
and the condition of Corollary~\ref{aequals1} is satisfied (take for instance
$\alpha = 3/4$ and $\const = 5/4$).
More generally, given values of $\alpha$ and $\const$ we can ask when 
there is a value of $a$ such that the skewed algorithm can certify
that most graphs on $n$ vertices
do not have any $\alpha n$-subgraph with excess $\const/\sqrt{n}$.
\begin{corollary} \label{skewcor}
If the condition $4 \alpha^2 \const^2 +5 \alpha^2 > 5$
is satisfied, the skewed algorithm can certify
that most graphs on $n$ vertices
do not have any $\alpha n$-subgraph with excess $\const/\sqrt{n}$.
\end{corollary}
\begin{proof}
In view of Theorem~\ref{skth}, we are looking for a value of $a$ 
such that $$f(a)=a^2+5/4-(\alpha \const + \alpha a)^2 < 0.$$
There is such an $a$ if $\Delta > 0$, where $\Delta$ is the discriminant 
of the quadratic polynomial $f(a)$. 
But $$\Delta = 4\alpha^4 \const^2 - 4(1-\alpha^2)(5/4-\alpha^2 \const^2) =
-5 + 4 \alpha^2 \const^2 + 5 \alpha^2.$$
\end{proof}
Note that this analysis of the skewed algorithm is 
suboptimal: in the case $a=0$ we can 
certify with the standard spectral algorithm that most graphs on $n$ vertices
do not have any $\alpha n$-subgraph with excess $\const/\sqrt{n}$ as soon
as $\alpha \const > 1$.
But this fact cannot be derived from Corollary~\ref{skewcor}.
In order to improve this corollary, one would need to improve the upper
bound on $\EE \norm{\sk}$ given by  Theorem~\ref{norm}.

It follows from Corollary~\ref{skewcor} that the constant $\kappa$ in Hypothesis~\ref{hyp:two} must satisfy the constraint $9\kappa^2 \leq 5$, i.e., 
$\kappa \leq \sqrt{5}/3 \simeq 0.745$.
Finally, let us consider the case where 
(as in Hypothesis~\ref{hyp:one})
the two functions $k(n)$ 
and $\epsilon(n)$ are of the form $k(n)=n^{\alpha}$ and 
$\epsilon(n)=1/n^{\beta}$ for two constants $\alpha,\beta \in ]0,1[$.
With those settings, the $\liminf$ in~(\ref{liminf}) is equal to
0 if $\alpha-\beta<1/2$, to 1 if $\alpha - \beta = 1/2$ and to $+\infty$ if
$\alpha-\beta > 1/2$. Hence we do not obtain any improvement over the
standard spectral algorithm (both algorithms succeed only in the case
$\alpha-\beta > 1/2$).
%
\section{Dense Subgraphs and the Restricted Isometry Property}
\label{hard_square}
%
In this section we show 
(in Theorems~\ref{hard_th1} and~\ref{hard_th2}) 
that RIP parameters are hard to approximate even for square matrices. 
We establish connections between dense subgraphs problems
and the RIP thanks to a generic reduction 
which we call the {\em Cholesky reduction}.
This reduction maps a graph $G$ on $n$ vertices to a $n \times n$ matrix
$C(G)$. Let $A$ be the signed adjacency matrix of $G$: we have
$a_{ii}=0$ and for $i \neq j$, $a_{ij}=1$ if $ij \in E$; 
$a_{ij}=-1$ if $ij {\not \in} E$.
We construct $C=C(G)$ from $A$ 
using the procedure described in Section~\ref{sectionC}.
That is, we first compute $B = \Id_n + c A/\sqrt{n}$ (suitable choices
for the parameter $c$ will be discussed later).
If $B$ is not positive semi-definite, we set $C=0$. Otherwise, we find
by Cholesky decomposition a matrix $C$ such that $C^TC=B$.

For suitable values of $k$, $C(G)$ satisfies the RIP of order $k$ for 
most graphs $G$. This was made precise in Proposition~\ref{modelC}.
On the other hand, we shall see that if $G$ has a $k$-subgraph that
is too dense then $C(G)$ does not satisfy the RIP of order $k$.
\begin{proposition} \label{excess}
Let $G$ be a graph on $n$ vertices and $C(G)$ its image by the Cholesky
reduction. If $G$ has a $k$-subgraph with excess $\epsilon$ 
and $$\delta<2c\epsilon(k-1)/\sqrt{n}$$ then $C(G)$ does not satisfy
the RIP of order $k$ with parameter $\delta$.
\end{proposition}
\begin{proof}
Let $A$ be the signed adjacency matrix of $G$ and $B = \Id_n + c A/\sqrt{n}$.
If $B$ is not semi-definite positive, $C(G)=0$ does not satisfy the RIP.
Otherwise $C^TC=B$. Let $x$ be the vector of Lemma~\ref{expansion}.
We have $||Cx||^2=x^TC^TCx=x^TBx=1+cx^TAx/\sqrt{n} > 1+\delta$.
\end{proof}
%
\subsection{Hardness from Hypothesis~\ref{hyp:one}}
%
\begin{lemma} \label{hard_lemma1}
 Set $k=n^{\alpha}$ and $\epsilon=1/n^{\beta}$ where
(as suggested in Section~\ref{subgraphs}) $\alpha, \beta$ are two constants
satisfying the constraints $\beta \in ]0,1/2[$ 
and $\alpha \in ]2\beta,\beta+1/2[$. 
Set also  $\delta=c'\epsilon k / \sqrt{n}$ where $c'>0$ is another constant.
Then for most graphs $G$ on $n$ vertices the matrix $C(G)$ satisfies
the RIP of order $k$ with parameter~$\delta$.
\end{lemma}
\begin{proof}
Fix a constant $c <1/3$.
We can apply Proposition~\ref{modelC} since the hypothesis 
$3c< \min (1,\delta\sqrt{n}/\sqrt{k})$ is satisfied for $n$ large enough
(note that 
$\delta \sqrt{n}/\sqrt{k} = c' \epsilon \sqrt{k} =c' n^{\alpha/2-\beta}$,
and the exponent $\alpha/2-\beta$ is positive).

In Proposition~\ref{modelC}, the probability that $C=C(G)$ does {\em not}
satisfy the RIP is bounded by a sum of two terms. The second one,
$4\exp[-(1/c -3)^2 n/32]$, clearly goes to 0 as $n \rightarrow + \infty$.
To obtain the same property for the first term we consider the argument:
\begin{equation} \label{argument}
k \ln (ne/k)-( \frac{\delta\sqrt{n}}{c} 
-3\sqrt{k})^2/32.
\end{equation}
Note that $\delta \sqrt{n}/c = (c'/c)\epsilon k = (c'/c)n^{\alpha-\beta}$, 
so that
this term dominates the term $3\sqrt{k}=3 n^{\alpha/2}$ 
since $\alpha-\beta > \alpha/2$.
Moreover $(\delta \sqrt{n}/c)^2=(c'/c)^2 n^{2(\alpha-\beta)}$ 
dominates the term
$k \ln(ne/k)$ since $k=n^{\alpha}$ and $2(\alpha-\beta) > \alpha$.
Therefore~(\ref{argument}) is equivalent to $-(c'/c)^2n^{2(\alpha-\beta)}/32$,
and its exponential goes to 0 as $n$ goes to  infinity.
Hence the same is true for the sum of the two error terms.
\end{proof}
\begin{theorem} \label{hard_th1}
Under Hypothesis~\ref{hyp:one}, no polynomial-time algorithm can 
approximate the RIP parameters of square matrices within any constant factor.

More precisely, fix any constants $\lambda > 1$ and $\delta_0 < 1/\lambda$.
Under Hypothesis~\ref{hyp:one}, no polynomial time algorithm 
can distinguish 
between matrices which satisfy the RIP of order $k$ with parameter 
$\delta \leq \delta_0$
and matrices which do not satisfy 
the RIP of order $k$ with parameter $\lambda \delta$.
\end{theorem}
In this theorem a distinguishing algorithm must accept all inputs
which satisfy the RIP of order $k$ with parameter 
$\delta \leq \delta_0$, and reject all inputs which do not satisfy
the RIP of order $k$ with parameter $\lambda \delta$. Its behavior on
other inputs is not specified.
\begin{proof}[of Theorem~\ref{hard_th1}]
The second part of the theorem clearly implies the first: if we have
an algorithm which approximates RIP parameters within some factor
$\gamma>1$ then we have a distinguishing algorithm with $\lambda = \gamma^2$.

Let us therefore assume that we have a distinguishing algorithm $\cal A$
with the associated constants $\lambda$ and $\delta_0$.
Fix a positive constant $c<1/3$ and two constants $\alpha$ and $\beta$ 
satisfying the constraints of Lemma~\ref{hard_lemma1},
for instance $\beta=1/3$ and $\alpha=3/4$.
Set as usual $k=n^{\alpha}$ and $\epsilon=1/n^{\beta}$.
We will use $\cal A$ to certify that most graphs on $n$ vertices do not have
any $k$-subgraph with excess $\epsilon$, thereby contradicting Hypothesis~\ref{hyp:one}.

On the one hand, by Lemma~\ref{hard_lemma1} for  most graphs $G$ 
the matrix $C(G)$ satisfy
the RIP of order $k$ with parameter $\delta=c'\epsilon k /\sqrt{n}$
(we will see quite soon that $c'$ can be any sufficiently small constant).
In this case $\delta = c'\epsilon k / \sqrt{n}=c'n^{\alpha-\beta-1/2}$
goes to 0 as $n \rightarrow +\infty$, so for $n$ large enough we have
$\delta \leq \delta_0$ and we can use the hypothesis on our distinguishing 
algorithm.

On the other hand, if $G$ has a $k$-subgraph with excess $\epsilon$ then 
by Proposition~\ref{excess} the RIP parameter of $C(G)$ is not smaller
than $2c\epsilon(k-1)/\sqrt{n}$. This is more than $\lambda \delta$
if $\lambda c' < 2c$ and $n$ is large enough. Therefore
given $\lambda$ and $c$ we choose a constant $c'<2c/\lambda$, and by running
$\cal A$ on $C(G)$ we can certify that most graphs on $n$ vertices do not have
any $k$-subgraph with excess $\epsilon$ (we certify $G$ if $\cal A$ accepts
$C(G)$).
\end{proof}
Note that the above hardness result has been established for {\em very small} 
RIP parameters: the matrices in the proof have RIP parameters of order
$\epsilon k / \sqrt{n}=n^{\alpha-\beta-1/2}$ 
and the exponent $\alpha-\beta-1/2$ is negative.
This may be viewed as a weakness of the result 
since applications to compressed sensing
only require a constant $\delta$ (as pointed out in the introduction,
any $\delta < \sqrt{2}-1$ will do).
We partially 
overcome this weakness in Section~\ref{hard2}, but for this we need 
to replace Hypothesis~\ref{hyp:one} by Hypothesis~\ref{hyp:two}.

Note also that the proof of Theorem~\ref{hard_th1} yields more information than contained in the statement of 
the theorem. For instance, with the settings $\alpha=3/4$ and $\beta=1/3$
we see that the gap problem in Theorem~\ref{hard_th1} remains hard 
even for $k=n^{3/4}$ and $\delta=c'/n^{1/12}$.
We have included this additional information in the table at the end of the
paper, for Theorem~\ref{hard_th1} as well as for our other hardness results.
%
\subsection{Hardness from Hypothesis~\ref{hyp:two}} \label{hard2}
%
\begin{lemma} \label{hard_lemma2}
 Set $k=\alpha n$ and $\epsilon=\const/\sqrt{n}$ where
$\alpha \in ]0,1[$ and $\const > 0$ are two constants.
Fix a constant $c < 1/3$ and another constant $c'$ such that 
$3c < c' \const \sqrt{\alpha}$  and 
\begin{equation} \label{c'}
\ln(e/\alpha)  < (\frac{c'\const\sqrt{\alpha}}{c}-3)^2/32.
\end{equation}
Finally, set $\delta=c'\epsilon k / \sqrt{n}=c'\alpha \const$.
Then for most graphs $G$ on $n$ vertices the matrix $C(G)$ satisfies
the RIP of order $k$ with parameter~$\delta$.
\end{lemma}
\begin{proof}
We can apply Proposition~\ref{modelC} since the hypothesis 
$3c< \min (1,\delta\sqrt{n}/\sqrt{k})$ is satisfied:
$\delta \sqrt{n}/\sqrt{k} = \delta / \sqrt{\alpha} = c' \const \sqrt{\alpha} > 3c$
by choice of $c'$.

As in the proof of Lemma~\ref{hard_lemma1}, the second error term 
in Proposition~\ref{modelC} clearly goes to 0 as $n \rightarrow + \infty$.
We therefore turn our attention to the first error term,
and in particular to the argument~(\ref{argument})
 of the exponential function.

In~(\ref{argument}) the term $\delta \sqrt{n}/c-3\sqrt{k}$ is equal to
$(c'\alpha \const /c - 3 \sqrt{\alpha})\sqrt{n}$ 
so the term $(\frac{\delta\sqrt{n}}{c} -3\sqrt{k})^2/32$ 
is equal to 
$\alpha(c'\sqrt{\alpha} A /c - 3 )^2n/32.$

The positive term in~(\ref{argument}) is $k \ln (ne/k)=\alpha n \ln(e/\alpha)$. Overall, $(\ref{argument})$ is equal to $-Kn$ for some constant $K$,
and $K$ is positive by~(\ref{c'}).
\end{proof}
\begin{theorem} \label{hard_th2}
Fix any constant $\delta_0 < \kappa/2$, where $\kappa$ is the constant
from Hypothesis~\ref{hyp:two}.
Assuming this hypothesis, no polynomial time algorithm 
can distinguish 
between matrices which satisfy the RIP of order $k$ with parameter $\delta_0$
and matrices which do not satisfy 
the RIP of order $k$ with parameter $\kappa/2$.
\end{theorem}
\begin{proof}
Let us assume the contrary, and let $\cal A$ be the distinguishing algorithm.
 Set $k=\alpha n$ and $\epsilon=\const/\sqrt{n}$ where $\const = \kappa/\alpha$ and
the constant 
$\alpha \in ]0,\kappa[$ is small enough (we will see at the end of the proof
how small is small enough).
We will use $\cal A$ to certify that most graphs on $n$ vertices do not have
any $k$-subgraph with excess $\epsilon$, thereby contradicting Hypothesis~\ref{hyp:two}.

Consider first the case of a graph $G$ having a $k$-subgraph 
with excess~$\epsilon$. 
Let us fix a constant $c<1/3$ such that $2c \kappa > \kappa/2$.
Then by Proposition~\ref{excess} the matrix $C(G)$ does not satisfy
the RIP of order $k$ with parameter $\kappa/2$.

On the other hand, if we can find a constant $c'$ which satisfies the hypotheses of Lemma~\ref{hard_lemma2} and such that 
$\delta=c'\alpha \const = \kappa c' < \delta_0$ then $C(G)$ will satisfy the RIP
of order $k$ with parameter $\delta_0$ for most graph $G$.
We will therefore be able to use $\cal A$ to certify that  most graphs on $n$ vertices do not have
any $k$-subgraph with excess~$\epsilon$.

It just remains to explain how to choose $c'$. We will simply pick any $c'$
such that $\kappa c' < \delta_0$ and will then check that the hypotheses of 
Lemma~\ref{hard_lemma2} are satisfied if $\alpha$ is small enough.
First we have to check the condition $3c < c'\const \sqrt{\alpha}$.
The right hand side is equal to $\kappa c'/\sqrt{\alpha}$ and will exceed
the left-hand side if e.g. $\sqrt{\alpha} < c' \kappa$ (recall that $3c<1$).
The other condition to be checked is~(\ref{c'}).
For the same reason, it will be satisfied for small enough~$\alpha$:
the right-hand side is equivalent to $(\kappa c'/c)^2/(32\alpha)$,
which dominates the logarithmic function on the left-hand side.
\end{proof}
%
%
\section{Hardness for Rectangular Matrices and Randomized Certification}
\label{hard_rect}
%
In this section we show that the RIP parameters of rectangular
matrices are hard to approximate. This is the case of interest
in compressed sensing. 
In a sense this was already done in Section~\ref{hard_square}:
we have shown that the special case of square matrices is already hard.
Nevertheless, it is of interest to know that the problems remains hard
for {\em strictly rectangular} matrices. This is what we do in this 
section. Proofs are essentially by reduction from the square case.
We begin with a simple lemma.
\begin{lemma} \label{diag_rip}
Consider a matrix $\Phi$ with the following block structure:
$$\Phi=\left(
\begin{array}{cc}
A & 0\\  0 & B
\end{array}\right).$$
This matrix satisfies the RIP of order $k$ with parameter $\delta$ 
if and only if the same is true for both $A$ and $B$.
\end{lemma}
\begin{proof}
For an input vector $x$ with the corresponding block structure $x=(u\ v)$
we have $||x||^2=||u^2||+||v||^2$ and $||\Phi x||^2 = ||Au||^2 + ||Bv||^2.$
Therefore, if $\Phi$ satisfies the RIP of order $k$ with parameter $\delta$ 
then the same is true for $A$ (take $v=0$ and $u$ $k$-sparse).
The same argument applies also to $B$.

Conversely, assume that $A$ and $B$ satisfy the RIP of order $k$
 with parameter $\delta$. Let $x=(u\ v)$ be a $k$-sparse vector.
We have 
$||\Phi x||^2 - ||x||^2 = (||\Phi u||^2 - ||u||^2)+(||\Phi v||^2 - ||v||^2).$
Both $u$ and $v$ must be $k$-sparse,
so the first term is bounded in absolute value by $\delta ||u||^2$ and the
second one by $\delta ||v||^2$. 
The result follows since $||u||^2+||v||^2 = ||x||^2$.
\end{proof}
\begin{theorem} \label{hard_rec1}
Fix any constants $\lambda > 1$ and $\delta_0 < 1/\lambda$.
Under Hypothesis~\ref{hyp:one}, no polynomial time algorithm 
can distinguish 
between matrices which satisfy the RIP of order $k$ with parameter 
$\delta \leq \delta_0$
and matrices which do not satisfy 
the RIP of order $k$ with parameter $\lambda \delta$.

Moreover, polynomial-time distinction between these two cases remains
impossible even for matrices of size $2n \times (n+N)$ where
 $N=n^{1+\epsilon_0}$. Here $\epsilon_0 \in ]0,1/2[$ is an absolute constant
(independent of $\lambda$ and $\delta_0$).
\end{theorem}
The first part of this theorem is just Theorem~\ref{hard_th1}.
The second part is new. Its proof is a refinement of the proof of 
Theorem~\ref{hard_th1} (and refers to it).
\begin{proof}[of Theorem~\ref{hard_rec1}]
From a graph $G$ on $n$ vertices we construct the matrix
$$C'(G)=\left(
\begin{array}{cc}
C(G) & 0\\  0 & B_n
\end{array}\right)$$
where $C(G)$ is as in the previous section and $B_n$ is a matrix with
good restricted isometry properties. 
Its role is to ensure the rectangular format that we need for $C'(G)$.
Our specific choice for $B_n$ is the matrix constructed in~\cite{BDFKK,BDFKKb}.
It is of size $n \times N$ where $N=n^{1+\epsilon_0}$, and it satisfies the
RIP of order $n^{{\frac1{2}}+\epsilon_0}$ with parameter $n^{-\epsilon_0}$.
Moreover, $B_n$ can be constructed deterministically in time polynomial in $n$.
Note that $C'(G)$ is of size $2n \times (n+N)$ as required in the statement
of Theorem~\ref{hard_rec1}.

Let us assume that we have a distinguishing algorithm $\cal A$ which
works for matrices of size $2n \times (n+N)$.
We will certify that a graph $G$ on $n$ vertices does
 not have any $k$-subgraph with excess $\epsilon$ if $\cal A$ accepts
$C'(G)$. 
As in the proof of Theorem~\ref{hard_th1}, for suitable
choices of $k$, $\epsilon$ and $\delta$ this will yield a contradiction 
with Hypothesis~\ref{hyp:one}. We now fill in the remaining details.

Instead of setting $\alpha=3/4$ and $\beta=1/3$
we now set $\alpha=\epsilon_0$; in order to satisfy the constraints
of Lemma~\ref{hard_lemma1} we need to take $\beta$ in the interval
$]\epsilon_0,\frac1{2}(\frac1{2}+\epsilon_0)[$.
For reasons that will soon become clear  we pick a $\beta$ which also satisfies
the additional constraint $\beta < 2 \epsilon_0$, and we set as usual
$k=n^{\alpha}$ and $\epsilon=1/n^{\beta}$.
It remains to show that we can use $\cal A$ to certify that most graphs
on $n$ vertices do not have any $k$-subgraph with excess $\epsilon$.

On the one hand, as explained in the proof of Theorem~\ref{hard_th1}
for most graphs $C(G)$ satisfies the RIP of order $k$ with parameter
$\delta=c' n^{\alpha-\beta-1/2}$ 
(the constant $c'$ is chosen as in that proof).
Remember also that $B_n$ satisfies the RIP of order $k$ with parameter
$n^{-\epsilon_0}$. Due to the choice $\beta < 2\epsilon_0$ we have
$n^{-\epsilon_0} < \delta$ for $n$ large enough, so by Lemma~\ref{diag_rip} 
for most graphs $G$ the matrix 
$C'(G)$ will satisfy the RIP of order $k$ with parameter $\delta$.

On the other hand, we have seen in the proof of Theorem~\ref{hard_th1}
that if $n$ is large enough and 
$G$ has a $k$-subgraph with excess $\epsilon$ then its RIP parameter
is larger than $\lambda \delta$. By Lemma~\ref{diag_rip} the same is true
of $C'(G)$.
\end{proof}
We can also give a hardness result for rectangular matrices based on
Hypothesis~\ref{hyp:two}. In fact we need a randomized version of this hypothesis 
which we state explicitly below.
\begin{hypothesis}[Randomized version of Hypothesis 2]\label{hyp:two:random}
There is a constant $\kappa \in ]0,1[$ such that the following holds.

Fix two arbitrary constants $\alpha \in ]0,1[$ and $\const >0$ with 
$\alpha \leq \kappa$ and 
$\alpha \const \leq \kappa$. 
Set $\epsilon(n)=\const/\sqrt{n}$ and $k(n)=\alpha n$ 
(again we only consider values of $n$ such that $k(n) \in \nn$).
Then no polynomial time {\em randomized} algorithm can certify that
most graphs on $n$ vertices do not have any $k(n)$-subgraph 
with excess $\epsilon(n)$.
\end{hypothesis}
In the above hypothesis, the hypothetical randomized algorithm would have
to satisfy the following properties:
\begin{itemize}
\item[(i)] For most graphs $G$ on $n$ vertices, 
with probability at least (say) 3/4 the algorithm certifies that $G$ 
does not contain any $k(n)$-subgraph with excess $\epsilon(n)$.
\item[(ii)] For all graphs $G$ on $n$ vertices, if $G$ is certified
then it is always true (with probability 1) that $G$ does not contain
any $k(n)$-subgraph with excess $\epsilon(n)$.
\end{itemize}
Note that the probability bounds in (i) and (ii) refer to the {\em internal}
 coin tosses of the algorithm.
\begin{theorem} \label{hard_rec2}
Fix any constant $\delta_0 < \kappa/2$, where $\kappa$ is the constant
from Hypothesis~\ref{hyp:two:random}.
Assuming this hypothesis, no polynomial time algorithm 
can distinguish 
between matrices which satisfy the RIP of order $k$ with parameter $\delta_0$
and matrices which do not satisfy 
the RIP of order $k$ with parameter $\kappa/2$.

Moreover, polynomial-time distinction between these two cases remains
impossible even for matrices of size $2n \times 100n$.
\end{theorem}
Again, only the second part of the theorem is new. 
\begin{proof}
As in the proof of Theorem~\ref{hard_rec1} we construct from a graph $G$
a matrix of the form
$$C'(G)=\left(
\begin{array}{cc}
C(G) & 0\\  0 & B_n
\end{array}\right).$$
For $B_n$, instead of of the deterministic construction from~\cite{BDFKK,BDFKKb} we will use a $n \times 99n$ random matrix given by Theorem~\ref{randrip}.
As before, we will certify that $G$  does
 not have any $k$-subgraph with excess $\epsilon$ if the hypothetical 
distinguishing algorithm accepts $C'(G)$. This will yield a contradiction
with the randomized version of Hypothesis~\ref{hyp:two}.

In the proof of Theorem~\ref{hard_th2} we showed how to set the various 
parameters so that $C(G)$ does not satisfy the RIP of order $k$ 
with parameter 1/2 in the case where $G$ has a $k$-subgraph with 
excess $\epsilon$. In this case,
by Lemma~\ref{diag_rip}
 it is also true that $C'(G)$ does not satisfy the RIP of order $k$ 
with parameter 1/2. Therefore
we will never certify a graph which contains a $k$-subgraph with 
excess $\epsilon$ (note that this holds irrespective of the choice of $B_n$).

We now turn our attention to the other case. By Theorem~\ref{randrip},
with probability approaching 1 as $n \rightarrow +\infty$ matrix
$B_n$ will satisfy the RIP of order $k$ with parameter $\delta_0$
if the condition
\begin{equation} \label{condition_hardrec2}
n \geq C \delta_0^{-2} k \log (eN/k).
\end{equation}
is satisfied. Since $k=\alpha n$ and $N=99n$ this is equivalent to
$\delta_0^2 \geq C \alpha \log(99e/\alpha).$
Note that the right-hand side of this expression goes to 0 with $\alpha$,
but the left-hand side is constant. Condition~\ref{condition_hardrec2}
will therefore be satisfied for any small enough $\alpha$.
But recall from the proof of Theorem~\ref{hard_th2} that for small enough
$\alpha$ we can set the other parameters so that $C(G)$ satisfies
the RIP of order $k$ with parameter $\delta_0$ for most graphs $G$.
With such a setting of parameters, we conclude from Lemma~\ref{diag_rip}
that $C'(G)$ satisfies the RIP of order $k$ with parameter $\delta_0$
for most graphs and most choices of the random matrix $B_n$.
\end{proof}
The constant 100 in Theorem~\ref{hard_rec2}
 can be replaced by any other constant.
Note also that the hypothetical polynomial-time algorithm in this theorem
remains deterministic: 
it is only the (hypothetical) algorithm for graph certification which
is randomized. 
It is clear, however, that Theorem~\ref{hard_rec2} remains
true for randomized algorithms with one-sided error.

Note that depending on the value of the constant $\kappa$ in Hypothesis~\ref{hyp:two:random}, 
Theorem~\ref{hard_rec2} does not rule necessarily rule out the existence
of a polynomial time algorithm for deciding whether a RIP parameter
is smaller than $\sqrt{2}-1$ (recall from the introduction that 
this is good enough for applications to compressed sensing).
Such an algorithm is ruled out if $\sqrt{2}-1 \in ]\kappa/2,\kappa[$, 
but could in principle exist if $\kappa < \sqrt{2}-1$. \footnote{From Section~\ref{sec:skewed_alg} we know that $\kappa \leq \sqrt{5}/3 \simeq 0.745$, 
so in any case we have $\kappa/2<\sqrt{2}-1$.}

Likewise, our hardness results
 do not rule out the existence of a polynomial-time algorithm 
distinguishing between matrices with a {very small}
RIP parameter and matrices with a RIP parameter larger than say 0.1.
Here {\em very small} means as in Theorem~\ref{hard_th1} that the
RIP parameter goes to 0 as $n \rightarrow +\infty$. 
If convergence to 0 is not too fast then we could still
 use such a weak distinguishing algorithm for certifying most random matrices. The following table gives a summary of our hardness results.
\begin{table}[ht] 
\small
	\begin{center}
    \begin{tabular}{ | c | p{5cm} | c | p{3cm} | c |}
	\hline
    \multicolumn{5}{|c|}{} \\
	\multicolumn{5}{|c|}{\underline{\textbf{
Hardness Results}}} \\

    \multicolumn{5}{|c|}{} \\
	\hline\hline
	\multirow{2}{*}{} & & & & \\
	$k$  & \text{$(k,\delta_1)$ vs. $(k, \delta_2)$ - hard }  & \textbf{Result} & \textbf{Assumptions} & \textbf{Dimensions ($n\times N$)} \\
	\hline
	\hline
	\multirow{2}{*}{} & & & & \\ 
	     $n^{3/4}$ & $\delta_1 = c_1n^{-1/12}$, $\delta_2 = \lambda c_1n^{-1/12}$ $\forall \lambda > 1$  &  Theorem~\ref{hard_th1} & Hypothesis~\ref{hyp:one} & $n\times n$ \\
	\hline
	\multirow{2}{*}{} & & & & \\ 
	     $n^{1/3}$ & $\delta_1 = c_1 n^{-4/15}$, $\delta_2 = \lambda c_1n^{-4/15}$ $\forall \lambda > 1$  &  Theorem~\ref{hard_th1} & Hypothesis~\ref{hyp:one} & $n\times n$ \\
	\hline
\multirow{2}{*}{} & & & & \\
    $\Theta(n)$ & $\delta_1 =\kappa / 2$, $\delta_2= \kappa$ &  Theorem~\ref{hard_th2}  & Hypothesis~\ref{hyp:two} & $n\times n $\\
\hline
\multirow{2}{*}{} & & & & \\
    $n^{3/4}$ & $\delta_1 = c_2 n^{-1/12}$, $\delta_2 = \lambda c_2 n^{-1/12}$ $\forall \lambda >1$  &   Theorem~\ref{hard_rec1}  & Hypothesis~\ref{hyp:one} & $n\times n^{1+\epsilon_0}$ for $\epsilon_0>0$ \\
\hline
\multirow{2}{*}{} & & & & \\
   $\Theta (n)$ & $\delta_1 =\kappa / 2$, $\delta_2= \kappa$ &   Theorem~\ref{hard_rec2}  &Hypothesis~\ref{hyp:two:random} & $n\times 50 n$ \\
\hline
\end{tabular}
\caption{We say that a matrix $\Phi$ has the $(k,\delta)$-RIP iff $(1-\delta) \leq \norm{\Phi x}^2 \leq (1+\delta)$ for every $k$-sparse unit vector $x$. By $(k,\delta_1)$ vs. $(k,\delta_2)$-hard we abbreviate the following: no polynomial time algorithm
can distinguish matrices $\Phi$ that satisfy the $(k,\delta_1)$-RIP 
from  matrices that do not satisfy the $(k,\delta_2)$-RIP. The constants $c_1,c_2$ depend on $\lambda$.}
\end{center}
\end{table}
\section{Acknowledgements}
We thank Eden Chlamt\'{a}\v{c} for many useful discussions. In particular, Eden suggested to study the complexity of the RIP for square matrices, and helped refute an earlier version of Hypothesis~\ref{hyp:two}.
%
%
%

\section*{Appendix}
\paragraph{Proof of Theorem~\ref{norm}}
We recall a bound on the expectation of a symmetric Bernoulli matrix~\cite[Eqn.~$1$]{random_matrix:norm:Vu}.
\begin{lemma}\label{lem:bound_exp}
Let $A$ be a symmetric matrix whose off-diagonal entries are i.i.d. from $\{\pm 1\}$ with equal probability and having zero diagonal entries. There exist positive constants $\delta,c_1$ and $c_2$ so that for $k=c_1 n^\delta$
\begin{equation}
\EE \trace{A^k} \leq c_2 n (2\sqrt{n})^k.
\end{equation} 
\end{lemma}
By definition of the operator norm, we know that $\norm{A} := \sup_{\norm{x}=1}{\norm{Ax}}$ for any matrix $A$. Decompose any unit vector as $x = \alpha \mathbf{1}/\sqrt{n} + \beta y$, where $\mathbf{1}$ is the all-ones vector, $y$ is the unit vector with $y \bot \mathbf{1}$ and $\alpha^2+\beta^2 =1 $. 
\begin{eqnarray}
	\EE \norm{\widehat{A}} &   =  & \EE \sup_{\norm{x}=1} \norm{\widehat{A}x}  
						   \   =  \ \EE  \sup_{\norm{y}=1,\ y\bot \mathbf{1}} \sup_{\alpha^2 + \beta^2 = 1} \norm{\widehat{A}(\alpha \frac{\mathbf{1}}{\sqrt{n}} + \beta y)}\nonumber  \\
						   & \leq & \EE  \sup_{\norm{y}=1,\ y\bot \mathbf{1}} \sup_{\alpha^2 + \beta^2 = 1} |\alpha|\norm{\widehat{A}\frac{\mathbf{1}}{\sqrt{n}}} + |\beta|\norm{\widehat{A}y}  
						   \ \leq \ \EE  \sup_{\norm{y}=1,\ y\bot \mathbf{1}} \sqrt{\norm{\widehat{A}\frac{\mathbf{1}}{\sqrt{n}}}^2 + \norm{\widehat{A}y}^2}\nonumber \\
						   &   =  & \EE  \sqrt{ \norm{\widehat{A}\frac{\mathbf{1}}{\sqrt{n}}}^2 + \sup_{\norm{y}=1,\ y\bot \mathbf{1}}\norm{\widehat{A}y}^2} 
						   \  \leq  \ \sqrt{\EE  \norm{\widehat{A}\frac{\mathbf{1}}{\sqrt{n}}}^2 + \sup_{\norm{y}=1,\ y\bot \mathbf{1}}\norm{\widehat{A}y}^2}\nonumber \\
						   &  \leq  & \sqrt{ \EE  \norm{\widehat{A}\frac{\mathbf{1}}{\sqrt{n}} }^2 + \EE \norm{A / 2}^2}\label{ineq:two_terms}
\end{eqnarray}
where we used the triangle's inequality, Cauchy-Schwarz's inequality, Jensen's inequality and the fact that $\J y =0$ for every $y\bot \mathbf{1}$. Now we bound each 
of the two terms on the right-hand side of~\eqref{ineq:two_terms} 
separately; the first term is straight-forward to bound. Indeed,
\begin{eqnarray*}
	\norm{\widehat{A}\frac{\mathbf{1}}{\sqrt{n}} }^2  &  =  &   \norm{A\frac{\mathbf{1}}{2\sqrt{n}} + a \mathbf{1} }^2 \  = \   \norm{A\frac{\mathbf{1}}{2\sqrt{n} } }^2 + 2\ip{ A\frac{\mathbf{1} }{2 \sqrt{n}} }{a \mathbf{1}} + a^2 n \\
							  &  =  &   \norm{A\frac{\mathbf{1}}{2\sqrt{n} } }^2 + (a/\sqrt{n})\sum_{i,j}A_{i,j} + a^2 n \  = \ \frac1{4n}\sum_{i}\left(\sum_{j} A_{ij}\right)^2  + (a/\sqrt{n})\sum_{i,j}A_{i,j} + a^2 n \\
							&  =  &   \frac1{4n}\sum_{i}\left(\sum_{j} A_{ij}^2 + 2\sum_{l<k} A_{il} A_{ik}   \right)  + (a/\sqrt{n})\sum_{i,j}A_{i,j} + a^2 n \\
\end{eqnarray*}
Now, take expectation and recall that $\EE A_{ij} = 0$ 
and $\EE A_{ij}^2 = 1$.
Since the entries of $A$ are independent, we obtain the inequality:
\begin{equation}\label{eqn:first_term}
\EE \norm{\widehat{A}\frac{\mathbf{1}}{\sqrt{n}} }^2 \leq  (a^2 + 1/4) n .\
\end{equation}
To 
bound the second term, we use Wigner's trace method (Lemma~\ref{lem:bound_exp}). For $k$ as in Lemma~\ref{lem:bound_exp},
\begin{eqnarray*}
 \EE{ \norm{A}^2} &  =  & \EE{\sqrt[k]{\norm{ A^{2k}}} } \ \leq \   \EE{\sqrt[k]{\trace{ A^{2k}}} } \ \leq \  \sqrt[k]{\EE{\trace{ A^{2k}}} } \ \leq\ 4n 2^{\OO (\log (n) /n^\delta ) }.
\end{eqnarray*}
where we used that $\norm{A^{2k}} \leq \trace{ A^{2k}}$, Jensen's inequality and Lemma~\ref{lem:bound_exp}.
The above discussion implies that
\begin{equation}\label{eqn:second_term}	
\EE{ \norm{A/2}^2}  \leq n (1+o(1)).
\end{equation}
Combining 
\eqref{ineq:two_terms},\eqref{eqn:first_term}
and 
\eqref{eqn:second_term}, it follows that
\[ \EE \norm{\widehat{A}}  \leq \sqrt{n} \sqrt{a^2 +5/4 + o(1)}. \]
%
%
%

This completes the proof of the first part of Theorem~\ref{norm}.
Now, we describe a tail bound for a convex function of random signs. This result was established by Ledoux~\cite{Ledoux96ontalagrand}; see also~\cite[\S~5.2]{book:Ledoux} for a discussion of concentration in product spaces.
\begin{theorem}
Suppose $f: \RR^m \to \RR $ is a convex function that satisfies the Lipschitz bound
\[|f(x)-f(y)| \leq L \norm{x-y}\quad \text{for all }x,y\in\RR^m.\]
Let $r\in{\{\pm 1\}^m}$ be a random sign vector. For all $\epsilon\geq 0$,
\[\Pr \left[ f(r) \geq \EE f(r) + L\epsilon \right] \leq \exp(-\epsilon^2 /8).\]
\end{theorem}
Let $f:\RR^{n(n-1)/2} \to \RR$. Consider the bijection between the set of symmetric matrices of size $n$ having zero diagonal and the set of vectors in $\RR^{n(n-1)/2}$ via the natural vectorization function $\text{vec}:\RR^{n\times n} \to \RR^{n(n-1)/2}$, i.e., vectorize the strict lower triangular part of its argument matrix. 
For every vector $v \in \RR^{n(n-1)/2}$,
define $f( v) := \norm{\frac1{2} \text{vec}^{-1}(v) + \frac{a}{\sqrt{n}} \J} $. By definition, $f$ is convex. Let's compute its Lipschitz constant; let $A,B$ be symmetric matrices of size $n$ with zero diagonals, then  $ \abs{f(\text{vec}(A)) -f(\text{vec}(B))} = \abs{\norm{\frac1{2} A + \frac{a}{\sqrt{n}} \J} - \norm{\frac1{2} B + \frac{a}{\sqrt{n}} \J} } \leq \norm{\frac1{2} (A-B)}\leq 1/2 \frobnorm{A-B} \leq \norm{ \text{vec}(A) - \text{vec}(B)}$ where $\frobnorm{A}^2 = \sum_{i,j}A_{ij}^2 $ (the inequality $\norm{A} \leq \frobnorm{A}$ always holds). Hence the Lipschitz constant of $f$ is one, which concludes the 
proof of the
``moreover'' part of Theorem~\ref{norm}.

\end{document}